\documentclass[conference]{IEEEtran}
\IEEEoverridecommandlockouts
\usepackage{cite}
\usepackage{amsmath,amssymb,amsfonts}
\usepackage{graphicx}
\usepackage{mypackages}
\usepackage{textcomp}
\usepackage{xcolor}
\def\BibTeX{{\rm B\kern-.05em{\sc i\kern-.025em b}\kern-.08em
    T\kern-.1667em\lower.7ex\hbox{E}\kern-.125emX}}
\begin{document}

\title{The Symmetric alpha-Stable Privacy Mechanism\\
% \thanks{This work was supported in part by Grant Number ECCS-2024493 from the U.S. National Science Foundation}
}

\author{\IEEEauthorblockN{ Christopher C. Zawacki}
\IEEEauthorblockA{\textit{Dept. of Electrical and Computer Engineering} \\
\textit{University of Maryland}\\
College Park, MD 20742 USA \\
czawacki@umd.edu}
\and
\IEEEauthorblockN{ Eyad H. Abed}
\IEEEauthorblockA{\textit{Dept. of Electrical and Computer Engineering} \\
\textit{University of Maryland}\\
College Park, MD 20742 USA \\
abed@umd.edu}
}

\maketitle

%%%%%%%%%%%%%%%%%%%%%%%%%%%%%%%%%%%%%%%%%%%%%%%%%%%%%%%%%%%%%%%%%%%%%%%%%%%%%%%%
\begin{abstract}
    With the rapid growth of digital platforms, there is increasing apprehension about how personal data is being collected, stored, and used by various entities. These concerns range from data breaches and cyber-attacks to potential misuse of personal information for targeted advertising and surveillance. As a result, differential privacy (DP) has emerged as a prominent tool for quantifying a system's level of protection. The Gaussian mechanism is commonly used because the Gaussian density is closed under convolution, a common method utilized when aggregating datasets. However, the Gaussian mechanism only satisfies approximate differential privacy. In this work, we present novel analysis of the Symmetric alpha-Stable (SaS) mechanism. We prove that the mechanism is purely differentially private while remaining closed under convolution. From our analysis, we believe the SaS Mechanism is an appealing choice for privacy focused applications. \newline
	\newline
	Tags: Differential Privacy, Stable Distributions, Data Privacy, Federated Learning \newline

\end{abstract}
%%%%%%%%%%%%%%%%%%%%%%%%%%%%%%%%%%%%%%%%%%%%%%%%%%%%%%%%%%%%%%%%%%%%%%%%%%%%%%%%

\section{INTRODUCTION}
\label{sec:introduction}
Privacy is fundamental to individual autonomy, rights, and personal safety. It protects individuals from harassment and discrimination, fosters trust in institutions, and encourages free speech and innovation. As the world becomes increasingly digital, we have seen a corresponding rise data breaches targeting the growing number of individual databases holding client information \cite{itrc22}. The public and private sectors have begun to act. Political leaders are taking actions to ensure the privacy of their citizens \cite{transparency11, GDPR18}, and consumers are putting pressure on companies to adopt settings and methods that focus on the privacy of their customers \cite{koetsier21,minto21}.

For example, in a healthcare context, an institution might want to disclose a histogram of blood glucose levels in a particular treatment group as evidence of a trials success. To prevent an adversarial agent from learning about which individuals make up the dataset, the institution can inject the dataset with differentially private noise. This provides bounds on the maximal amount of privacy that could be lost while allowing the results of the study to be made public. 

Differential Privacy (DP) is a method of noise injection that allows for quantifiable guarantees about the amount of information that can be leaked when an individual participates in a machine learning dataset \cite{dwork06, dwork06b}. By perturbing datasets with random noise from a carefully selected density, DP ensures that the statistical results of any analysis remain accurate while protecting the identity and sensitive information of participating clients. The differential privacy framework has been applied in various domains, from large-scale data analytic to machine learning \cite{abadi16}. More recently, differential privacy has had renewed focus within the field of federated learning (FL), a privacy focused branch of machine learning \cite{mcmahan16}. The objective of differentially private FL methods are to enhance privacy preservation while collaboratively training machine learning models across multiple decentralized devices or servers \cite{wei20}. In \cite{li19b}, the authors use differentially private federated learning methods to train a machine learning model that segments images of brain tumors. 

The work most similar to the results presented here are from Ito et. al. \cite{ito21}, who use heavy tailed distributions to mask contributions by outliers in the context of filter/controller design for control systems. Our results differ in the level of privacy guaranteed by the privacy mechanism.

The contributions of this work are twofold. First we present a privacy mechanism that uses stable densities and prove that this mechanism is $\varepsilon$-differentially private. Second, we compare the expected distortion of our privacy mechanism against other commonly utilized privacy mechanisms.

The rest of the paper is organized as follows. Section \ref{sec:background} summarises the basics of differential privacy. Section \ref{sec:sas} introduces the definition of the Symmetric alpha Stable Mechanism. Section \ref{sec:priv} proves the privacy guarantee of the new mechanism. Section \ref{sec:error} provides a measure of error the mechanism introduces to statistical queries. Lastly, section \ref{sec:conclusion} summarizes the results and comments on active research efforts.

\section{BACKGROUND}
\label{sec:background}

In this section, we outline the background material required to derive our results. 

\subsection{Differential Privacy}
Differential privacy is a method of obfuscation that operates on a collaboratively constructed dataset $\mathcal{D}$ \cite{dwork06, roth14}. It is common to consider such a dataset as a tabulated set of records, where each row holds a vector of client data. 

Let $f$ be a function that operates on a dataset and returns a vector of $m$ numerical values. For example,
\begin{itemize}
    \item How many clients have blue eyes?
    \item What is the average income?
    \item What are the optimized parameters of a given machine learning model over all the clients?
\end{itemize}
By a slight abuse of notation, we use the same symbol $f$ for the function regardless of the size of the dataset.
\begin{definition}
\label{def:bquery}
(Bounded Query)
We call a function $f$ a bounded query if it takes as input a dataset $\mathcal{D}$ and returns a vector, of positive dimension $m$, taking values in compact subsets of the real line:
\begin{equation}
    f: \mathcal{D} \to [a_i, b_i]^m, \ a_i,b_i \in \mathbb{R},
\end{equation}
with $i \in \{1, 2, ..., m\}$.
    \vspace{-5mm}
    \begin{equation*}\tag*{\textrm{$\blacktriangleleft$}}\end{equation*}
    % \vspace{-3mm}
\end{definition}

More commonly, a query $f$ is allowed to be unbounded with differential privacy methods restricting the queries to those with finite $\ell_p$-sensitivity \cite{dwork06, roth14}:
\begin{definition}
\label{def:lp}
($\ell_p$-Sensitivity of Query) The $\ell_p$-sensitivity of a query $f$, denoted $\Delta_pf$, is defined to be a maximum of a $p-$norm over the domain of $f$, $dom(f)$:
\begin{equation}
    \Delta_pf := \max_{\mathcal{D}_1\simeq \mathcal{D}_2} ||f(\mathcal{D}_1) - f(\mathcal{D}_2)||_p,
\end{equation}
for all $\mathcal{D}_1, \mathcal{D}_2 \in dom(f)$.
\vspace{-4mm}
\begin{equation*}\tag*{\textrm{$\blacktriangleleft$}}\end{equation*}
\end{definition}
It is evident from Definitions \ref{def:bquery} and \ref{def:lp} that if and only if a query $f$ is bounded, the sensitivity of $f$ is also bounded. It simplifies our analysis to assume the query is bounded, which is tantamount to the assumption in the literature that the $\ell_p$-sensitivity is bounded.

\begin{definition}
(Privacy Mechanism) A privacy mechanism for the query $f$, denoted $\mathcal{M}_f$, is defined to be a randomized algorithm that returns the result of the query perturbed by a vector of pre-selected i.i.d. noise variables $Y_i$,
    \begin{equation}
        \mathcal{M}_f(\mathcal{D}) = f(\mathcal{D}) + (Y_1, Y_2, \dots, Y_m)^T,
    \end{equation}
    for all $i \in \{1, \dots, m\}$.
    \vspace{-5mm}
    \begin{equation*}\tag*{\textrm{$\blacktriangleleft$}}\end{equation*}
    % \vspace{-3mm}
\end{definition}
It will be useful to denote the vector $\mathcal{M}_f(\mathcal{D})$ as $\x\in\mathbb{R}^m$. Note that the noise variables, $Y_i$, induce a density, which we denote $p = p(\x)$ for $\mathcal{M}_f$, on a given dataset $\mathcal{D}$. While not strictly necessary, we assume the injected density is symmetric about the origin. This assumption simplifies the analysis.

Let us now consider two possible groups of clients depicted in Figure 1. In one scenario, the red client has decided to included their data in the dataset and, in the other, the red client chooses to withhold their data. Let $\mathcal{D}_1$ and $\mathcal{D}_2$ represent these two scenarios respectively . To proceed, let us assume the red client has allowed their data in the set and so $\mathcal{D}_1$ is the \textit{true} dataset. Denote a realization of a mechanism as $x \sim \mathcal{M}_f(\mathcal{D}_2)$. Informally, the mechanism $\mathcal{M}_f$ is said to be differentially private if the inclusion or exclusion of a single individual in the dataset, illustrated in red in the figure, results in \textit{roughly} the same distribution over the realized outputs $x$,
\begin{equation}
    \Pr[\mathcal{M}_f(\mathcal{D}_1) = x] \approx \Pr[\mathcal{M}_f(\mathcal{D}_2) = x].
\end{equation}
\begin{figure}[ht!]
	\centering
	\includegraphics[width=5cm]{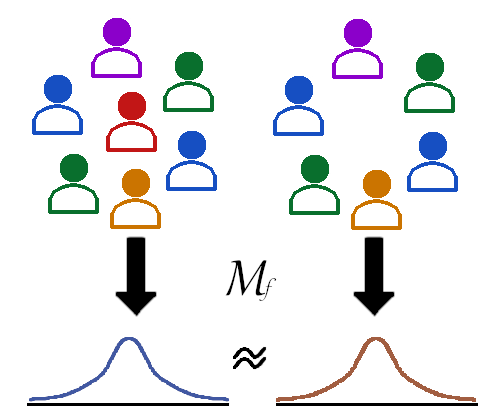}
	\caption{\label{fig:dp} Differential privacy quantifies the expected shift in distribution from the inclusion or exclusion of a single individual. In this figure, the red client is deciding whether to allow their data to be used in a collaborative dataset. If the query for this dataset is differentially private, then the expected result of the query will be \textit{approximately} the same regardless of their decision.}
\end{figure}
The privacy mechanism aims to hinder an adversary from conclusively ascertaining the presence of the \textit{red client} within the dataset. Next, we proceed to quantify this intuition.

\begin{definition}
    (Neighboring Datasets) Two datasets, denoted by $\mathcal{D}_1$ and $\mathcal{D}_2$, are known as neighboring datasets if they differ in the presence or absence of exactly one client record. We denote this relation as $\mathcal{D}_1\simeq \mathcal{D}_2$.
        \vspace{-7mm}
    \begin{equation*}\tag*{\textrm{$\blacktriangleleft$}}\end{equation*}
    % \vspace{-3mm}
\end{definition}

\begin{definition}
    (Pure Differential Privacy)
    Let $\mathcal{D}_1$ and $\mathcal{D}_2$ be any neighboring datasets. Given a query $f$ that operates on $\mathcal{D}_1$ and $\mathcal{D}_2$, a privacy mechanism $\mathcal{M}_f$ is said to be $\varepsilon$-differentially private ($\varepsilon$-DP) if it satisfies
    \begin{equation}
    \label{eqn:pureDP}
     \Pr[\mathcal{M}_f(\mathcal{D}_1) \in \mathcal{X}] \leq e^\varepsilon \Pr[\mathcal{M}_f(\mathcal{D}_2) \in \mathcal{X}]
    \end{equation}
    for some $\varepsilon > 0$ and any subset of outputs $\mathcal{X} \subseteq \mathcal{R}(\mathcal{M}_f(\mathcal{D}_1))$. The mechanism is defined to have no privacy ($\epsilon = \infty$) if, upon its application to each dataset, the supports of the resulting densities are not equal, viz. $\mathcal{R}(\mathcal{M}_f(\mathcal{D}_1)) \neq \mathcal{R}(\mathcal{M}_f(\mathcal{D}_2))$.
    \vspace{-5mm}
    \begin{equation*}\tag*{\textrm{$\blacktriangleleft$}}\end{equation*}
    % \vspace{-3mm}
    \end{definition}
\begin{remark}
    We note that Eq. \ref{eqn:pureDP} holds for each element when the density of the distributions is considered \cite{roth14}:
        \begin{equation}
     \label{eqn:purepdf}
        p_1(x) \leq e^\varepsilon p_2(x), \quad \forall x \in \mathcal{R}(\mathcal{M}_f(\mathcal{D}_1))
    \end{equation}
\end{remark}

The parameter $\varepsilon$ is also referred to as the privacy budget. Smaller values of $\varepsilon$ are, in general, associated with stronger privacy. We remark that when $\varepsilon=0$, the definition yields perfect privacy. However, in that case, adding more client data results in no new information.
% \begin{theorem}
% \label{thm:atoms} (Privacy as Densities)
%     Let $\mathcal{D}_1$ and $\mathcal{D}_2$ be neighboring datasets and $f$ be a query that operates on them. Denote by $p_1$ and $p_2$ the densities of the privacy mechanism $\mathcal{M}_f$ when applied to $\mathcal{D}_1$ and $\mathcal{D}_2$ respectively. Then, a privacy mechanism $\mathcal{M}_f$ is $\varepsilon$-differentially private if
%     \begin{equation}
%     \label{eqn:purepdf}
%         p_1(x) \leq e^\varepsilon p_2(x), \quad \forall x \in \mathcal{R}(\mathcal{M}_f(\mathcal{D}_1))
%     \end{equation}
%     for all $\mathcal{D}_1 \simeq \mathcal{D}_2$. 
%     \vspace{-5mm}
%     \begin{equation*}\tag*{\textrm{$\blacktriangleleft$}}\end{equation*}
%     % \vspace{-3mm}
% \end{theorem}
% Next, we give a brief proof for this known fact.
% \begin{proof}
%     Begin by writing condition (\ref{eqn:pureDP}) in terms of the generated densities,
%     \begin{equation}
%     \label{eqn:thm1-a}
%         \int_\mathcal{X} p_1(x) dx \leq \int_\mathcal{X} e^\epsilon p_2(x)dx.
%     \end{equation}
%     Equation (\ref{eqn:thm1-a}) can be rewritten as
%     \begin{equation}
%     \label{eqn:thm1-b}
%         0 \leq \int_\mathcal{X} e^\epsilon p_2(x) - p_1(x) dx.
%     \end{equation}
%     Noting that (\ref{eqn:purepdf}) enforces the integrand in (\ref{eqn:thm1-b}) to be non-negative implying that (\ref{eqn:pureDP}) is satisfied.
% \end{proof}

\begin{definition}   
(Privacy Loss) The privacy loss of an outcome $x$ is defined to be the log-ratio of the densities when the mechanism is applied to $\mathcal{D}_1$ and $\mathcal{D}_2$ at $x$ \cite{roth14}:
\begin{equation}
\label{eqn:plf}
 \mathcal{L}_{\mathcal{D}_1 || \mathcal{D}_2}(x) := \ln \frac{p_1(x)}{p_2(x)}.
\end{equation}
By (\ref{eqn:purepdf}), it is evident that $\varepsilon$-differential privacy (\ref{eqn:pureDP}) is equivalent to
\begin{equation}
    |\mathcal{L}_{\mathcal{D}_1||\mathcal{D}_2}(x)| \leq \varepsilon, \quad \forall x \in \mathcal{R}(\mathcal{M}_f(\mathcal{D}_1))
\end{equation}
for all neighboring datasets $\mathcal{D}_1$ and $\mathcal{D}_2$. 
\vspace{-7mm}
\begin{equation*}\tag*{\textrm{$\blacktriangleleft$}}\end{equation*}
% \vspace{-3mm}
\end{definition}
For mechanisms that are purely differential private, the privacy budget $\varepsilon$ is the maximum over all observations $x$,
\begin{equation}
\label{eqn:epmax}
    \varepsilon = \max_{x\in\mathbb{R}}\mathcal{L}_{\mathcal{D}_1||{D}_2}(x).
\end{equation}

Some mechanisms, such as the Gaussian mechanism \cite{dwork06b, roth14}, fail to satisfy condition (\ref{eqn:pureDP}). The condition can be relaxed through the inclusion of an additive constant $\delta > 0$, as in the following definition:
\begin{definition} (Approximate Differential Privacy)
    Let $\mathcal{D}_1$ and $\mathcal{D}_2$ be any neighboring datasets. Given a query $f$ that operates on $\mathcal{D}_1$ and $\mathcal{D}_2$, a privacy mechanism $\mathcal{M}_f$ is said to be $(\varepsilon, \delta)$-differentially private if it satisfies
\begin{equation}
\label{eqn:approxDP}
 \Pr[\mathcal{M}_f(\mathcal{D}_1) \in \mathcal{X}] \leq e^\varepsilon \Pr[\mathcal{M}_f(\mathcal{D}_2) \in \mathcal{X}] + \delta.
\end{equation}
This is known as approximate differential privacy.
\vspace{-4mm}
\begin{equation*}\tag*{\textrm{$\blacktriangleleft$}}\end{equation*}
% \vspace{-3mm}
\end{definition}
% In terms of the privacy loss random variable, approximate differential privacy is equivalent to
% \begin{equation}
%     \Pr\big[ |\mathcal{L}_{\mathcal{D}_1 || \mathcal{D}_2}(z \sim \mathcal{M}_f(\mathcal{D}_1)) | \geq \varepsilon\big] \leq\delta.
% \end{equation} 

Commonly, a mechanism is defined in relation to a query over the entire dataset $\mathcal{D}$. It is then understood that the mechanism is applied by a \textit{trusted aggregator}, which collects the clients' data prior to obfuscation. However, there does not always exist such a trusted central authority. For example, in a federated learning framework, the server is assumed untrustworthy by default. Moreover, a lack of secure communication protocols could result in an adversary gaining access to the transmission between the clients and server. To this end, a mechanism $\mathcal{M}_f$ is said to be locally differentially (LDP) private if the mechanism can be applied locally by the clients prior to transmission to the server.

\begin{definition} (Local Differential Privacy)
\label{def:local}
A privacy mechanism $\mathcal{M}^{loc}_f$ is said to be locally differentially private if, when applied to a client's local dataset $\mathcal{D}$,  satisfies for any pair of datapoints $v_1, v_2 \in \mathcal{D}$ the following \cite{kasiviswanathan11}:
\begin{equation}
\label{eqn:LDP}
 \Pr[\mathcal{M}^{loc}_f(v_1) \in \mathcal{X}] \leq e^\varepsilon \Pr[\mathcal{M}^{loc}_f(v_2) \in \mathcal{X}] + \delta,
\end{equation}
 for all $\mathcal{X} \in \mathcal{R}(\mathcal{M}^{loc}_f)$.
\vspace{-4mm}
\begin{equation*}\tag*{\textrm{$\blacktriangleleft$}}\end{equation*} 
% \vspace{-3mm}
\end{definition}
The mechanism is called $\varepsilon$-LDP if $\delta=0$ and $(\varepsilon, \delta)$-LDP otherwise.

\subsection{Selecting a Level of Privacy}
Wasserman and Zhou describe in \cite{wasserman09} a useful connection between differential privacy and hypothesis testing. Their analysis considers the problem of client privacy from the perspective of an adversary deciding between two hypothesises. Denote by $\mathcal{D}_1$ and $\mathcal{D}_2$ two neighboring datasets. Let one of the following hypothesises hold:
\begin{itemize}
    \item $H_0$ (The null hypothesis): the true dataset is $\mathcal{D}_1$.
    \item $H_1$ (The alternative hypothesis): the true dataset is $\mathcal{D}_2$.
\end{itemize}
The objective of the adversary is to determine, based on the output of a privacy mechanism $\mathcal{M}_f$, which hypothesis is true. Denote by $p$ the probability of a false positive, that is, the adversary chooses $H_1$ when $H_0$ is true. Then, denote by $q$ the probability of a false negative, i.e., $H_0$ is chosen when $H_1$ is true. Wasserman and Zhou show that if a mechanism $\mathcal{M}_f$ is $\varepsilon$-differentially private, then the following two statements must hold:
\begin{equation}
\label{eqn:h0}
    p + e^\varepsilon q \geq 1 \textrm{ and }  e^{\varepsilon}p + q \geq 1.
\end{equation}
% and 
% \begin{equation}
% \label{eqn:h1}
   
% \end{equation}
Combining the inequalities in (\ref{eqn:h0}) yields
\begin{equation}
\label{eqn:h2}
    p + q \geq \frac{2}{1 + e^\varepsilon}.
\end{equation}
Consider that when $\varepsilon << 1$, which equates to high privacy, the adversary cannot achieve both low false positive and low false negative rates simultaneously. Often, it is more convenient to specify lower bounds for $p$ and $q$ and to use (\ref{eqn:h2}) to determine $\varepsilon$ than it is to state the privacy budget directly.

\section{The Symmetric alpha-Stable Mechanism}
\label{sec:sas}
The Gaussian density constitutes one of the main privacy mechanisms in differential privacy. One major benefit is the ease with which Gaussian perturbations fit into existing Machine Learning analyses. The Gaussian density owes its pervasiveness to its essential role in the Central Limit Theorem (CLT) \cite[Thm. 27.1]{billingsley95}. One important property is that the density is closed under convolutions. This means two Gaussian estimates can be combined and the result remains Gaussian. 

We note that the Gaussian density is a member of a family of densities with this property, known as the L\'evy alpha-Stable density \cite{levy25}. In the context of Differential Privacy, the Gaussian mechanism only satisfies condition (\ref{eqn:approxDP}), approximate differential privacy \cite{dwork06b,roth14}. In this section, we examine the privacy properties of the mechanisms based on the family of stable densities. We introduce the Symmetric alpha-Stable mechanism and prove it satisfies condition (\ref{eqn:pureDP}), $\varepsilon$-DP.

\subsection{The Family of Stable Densities}
The family of stable densities was first studied in generality by L\'evy in 1925 \cite{levy25} and is defined to be the class of probability densities that are closed under convolution.
\begin{definition} (The Stable Family)
    Let $Y_1$ and $Y_2$ be two independent and identically distributed random variables following probability density $Y$. The density $Y$ is said to be \textit{stable} if for any constants $a, b > 0$ there exist constants $c(a,b) > 0$ and $d(a,b) \in \mathbb{R}$ such that
    \begin{equation}
        aY_1 + bY_2 = cY+d.
    \end{equation}
    If $d=0$, the distribution is known as \textit{strictly stable}.
    \vspace{-7mm}
    \begin{equation*}\tag*{\textrm{$\blacktriangleleft$}}\end{equation*}
    \vspace{-3mm}
\end{definition}
Aside from a few special cases, there is no known closed form for the density of a general stable density \cite{nolan20}. However, there are several known parameterizations of the characteristic function of a density in the stable family \cite{nolan20}. One common form  of the characterised function is
\begin{equation}
\label{eqn:char}
    \varphi(t; \alpha, \beta, \gamma, \mu) = \exp({it\mu - |\gamma t|^\alpha + i \beta \textrm{sgn}(t) \Phi(t)}),
\end{equation}
with
\begin{equation}
    \Phi(t) = \begin{cases}
        \tan(\frac{\pi \alpha}{2}) & \alpha \neq 1 \\
        -\frac{2}{\pi}\log|t| & \alpha = 1.
    \end{cases}
\end{equation}
The density can then be expressed by the integral
\begin{equation}
    \label{eqn:stab}
    p(x; \alpha, \beta, \gamma, \mu) = \frac{1}{2\pi}\int_{-\infty}^{\infty}\varphi(t; \alpha, \beta, \gamma, \mu)e^{-ixt}dt.
\end{equation}
We present three example of the symmetric form, with $\beta=0$, in Figure \ref{fig:sas_density}: $\alpha=1$ (blue), $\alpha=1.5$ (orange), and $\alpha=2$ (green). Each graph is standardized with a location of $\mu=0$ and a scale of $\gamma=1$.
\begin{figure}[ht!]
	\centering
	\includegraphics[width=\textwidth/2]{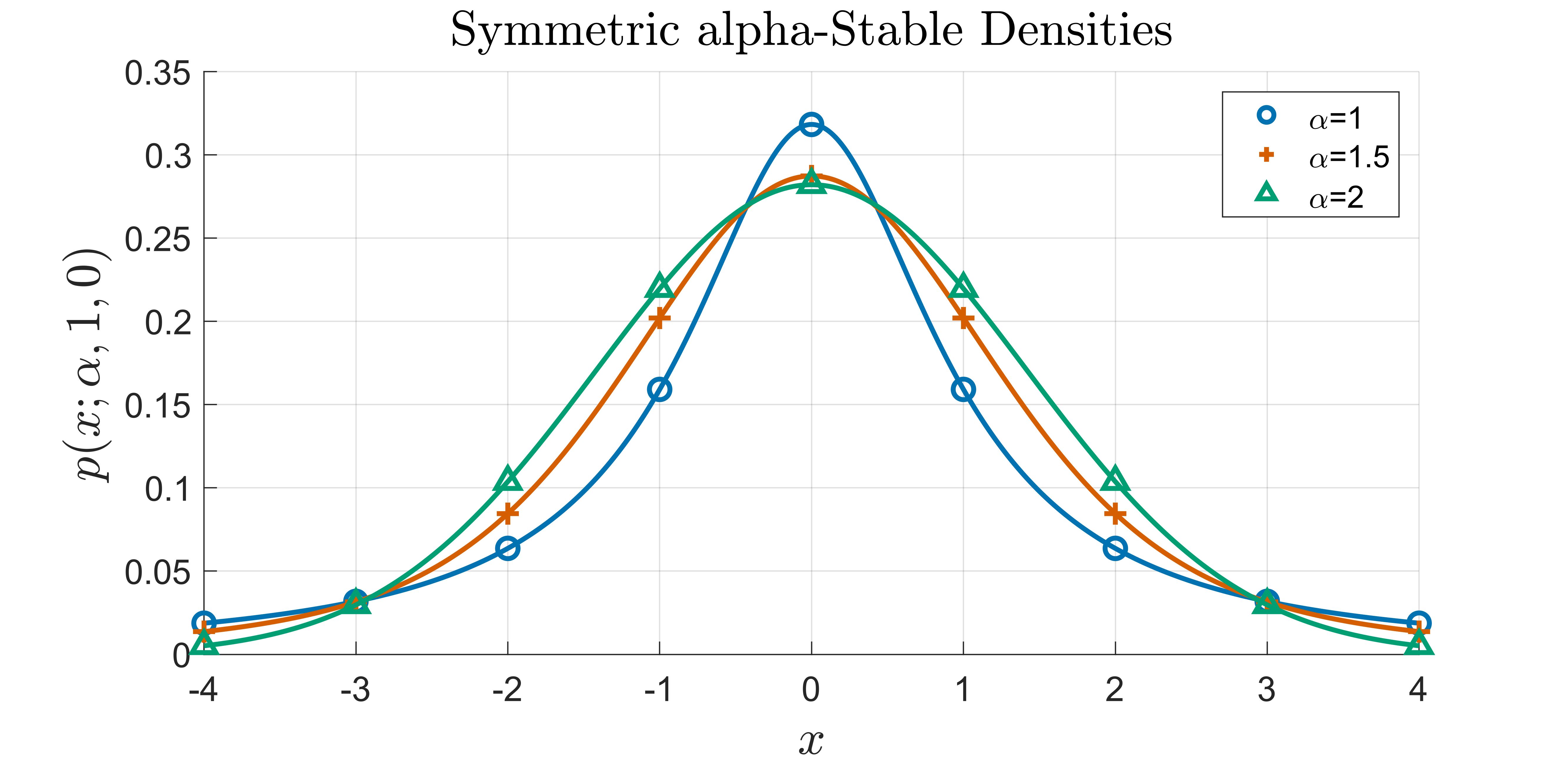}
	\caption{\label{fig:sas_density} The family of Symmetric alpha-Stable densities consists of bell shaped densities with varying rates of decay in the tail that are closed under convolutions. This figure depicts three densities each with zero mean, unit scale, and $\alpha=1$ in blue, $\alpha=1.5$ in orange, and $\alpha=2$ in green.}
\end{figure}
The two forms, $\alpha=1$ and $\alpha=2$, are the Cauchy and Gaussian densities respectively.

% The largest challenge when working with the family of stable distributions is the lack of a closed form solution for the general density. This is, in part, due to the fact that the value at an individual point is given by the integral of an infinitely oscillating function. 
% Denote the real part of the integrand of (\ref{eqn:stab}) by $q(t;x)$, then a depiction with $x=10$ can be found in Figure \ref{fig:slice_sas_density}.
% \begin{figure}[ht!]
% 	\centering
% 	\includegraphics[width=6cm]{figures/slice_of_density.png}
% 	\caption{\label{fig:slice_sas_density} This graph depicts the real part of the integrand of (\ref{eqn:stab}) (which is an oscillating function) for $\alpha=1.5$, $\gamma=1$, and $\mu=0$. The value of the stable density at a given point $x$ is the integral of this oscillating function. }
% \end{figure}

For (\ref{eqn:stab}) to constitute a probability density, the parameter $\alpha$ is constrained to lie within the interval $(0,2]$. The value of $\alpha$ determines the rate of decay of the tail of the density. The mean of the density is undefined for $\alpha \leq 1$ and defined for $\alpha > 1$. The density has infinite variance for $\alpha \in (0,2)$ and finite variance only when $\alpha=2$.  In this work, we restrict $\alpha \in (1, 2]$, leaving median or mode estimators for future research. The parameter $\beta$, restricted to $(-1,1)$, is a measure related to skewness (the strict definition of skewness is not meaningful for $\alpha < 2$). We focus on symmetric alpha-stable (SaS) densities which are defined to be the case where $\beta=0$:
\begin{equation}
\label{eqn:sas}
\begin{aligned}
    p_{SaS}(x;\alpha, \gamma, \mu) & := \\ p(x;\alpha, 0, \gamma, \mu) & = \frac{1}{2\pi}\int_{-\infty}^{\infty}e^{-|\gamma t|^\alpha - it(x-\mu)}dt.
\end{aligned}
\end{equation}
SaS densities have a known closed form for two values of the parameter $\alpha$: the Cauchy, for $\alpha = 1$, and the Gaussian, for $\alpha=2$. The last two parameters, $\gamma > 0$ and $\mu\in\mathbb{R}$, are the scale and location parameters respectively. 
\begin{remark}
    For stable densities, it is common for the location parameter to be denoted $\delta$ rather than $\mu$, to signify that it is not always equal to the expected value. In our context, we choose $\mu$ to reserve $\delta$ for the definition of approximate differential privacy (\ref{eqn:approxDP}) as is common in the differential privacy literature. Because we are restricting the domain of interest to $\alpha \in (1,2]$, we do not believe this notation will be cause for confusion.
\end{remark}

% \section{}

\begin{definition}
\label{def:sas}
    (The Symmetric alpha-Stable Mechanism)
    For a given dataset $\mathcal{D}$ and a query function $f$, we define a privacy mechanism $\mathcal{M}_f$ to be a Symmetric alpha-Stable (SaS) mechanism if each element of the vector of injected values, $Y_i$ for $i\in\{1,...,m\}$, is drawn independently from a SaS density
    \vspace{-3mm}
    \begin{equation}
        p_{SaS}(x; \alpha, \beta, 0) = \frac{1}{2\pi}\int_{-\infty}^{\infty}e^{-|\gamma t|^\alpha - itx}dt.
    \end{equation}
    \vspace{-7mm}
\begin{equation*}\tag*{\textrm{$\blacktriangleleft$}}\end{equation*}
\end{definition}

In this section, we proceed to prove that the SaS Mechanism for $\alpha\in[1,2)$ satisfies (\ref{eqn:pureDP}), $\varepsilon$-differential privacy.

\section{Pure-Differential Privacy of SaS Mechanism}
\label{sec:priv}
The main difficulty in working with stable densities, other than the Cauchy and Gaussian, is that they have no known closed form and consist of the integral of an infinite sequence of oscillating intervals.  In this section, we first establish the following lemmas and then we prove that the SaS Mechanism is $\varepsilon$-DP when $\alpha$ is restricted to $[1,2)$. 

To establish that the privacy loss is finite on a compact set, it is essential to ensure that the stable distribution possesses support over the entire real number line.
\begin{lemma}
\label{lem:sup}
    (Support of SaS Density)
    The support of the symmetric alpha-stable density (\ref{eqn:sas}) is $\mathbb{R}$.
\end{lemma}
\begin{proof}
See \cite[Lemma 1.1]{nolan20}.
\end{proof}

Next, we recall a partial sum expansion described by Bergstr\"om \cite{bergstrom52} where the remainder term has a smaller order of magnitude (for large $|x|$) then the final term in the series.
\begin{lemma}
\label{lem:sum}
    (Finite Series Expansion of SaS Distribution)
    The symmetric alpha-stable density (\ref{eqn:sas}), with $\alpha \in [1,2)$ and $\gamma=1$, has the following finite series expansion:
    \begin{equation}
        \label{eqn:series}
        \begin{aligned} 
            & p_{SaS}(x; \alpha, 1, 0) =  \\ & -  \frac{1}{\pi} \sum_{k=1}^n (-1)^k \frac{\Gamma(\alpha k+1)}{(x)^{\alpha k + 1}} \sin \bigg( \frac{k \alpha \pi}{2} \bigg) +  O\bigg(x^{-\alpha(n+1)-1}\bigg),
        \end{aligned}
    \end{equation}
    for $|x| \to \infty$.
\end{lemma}
\begin{proof}
    The proof provided by Bergstr\"om \cite{bergstrom52} employs an expanded form that is satisfied for the full range $\beta \in (-1, 1)$. We only require (\ref{eqn:series}) and so leave out the full form.
\end{proof}

We use the foregoing lemma to argue that the privacy loss remains bounded as the observation $|x|$ tends to infinity. However, Eq. (\ref{eqn:series}) is stated for $\gamma=1$. The next lemma states that the asymptotic behavior of the privacy loss as $|x|\to\infty$ is independent of $\gamma$.

% First, we denote a simplified notation of the privacy loss function (\ref{eqn:plf}), with SaS densities (\ref{eqn:sas}) equipped with scale parameter $\gamma$, as $\mathcal{L}^{SaS}(x; \gamma)$, suppressing the datasets.
\begin{lemma}
\label{lem:scale}
(No Scale Dependence in the Limit)
Let $\mathcal{D}_1 \simeq \mathcal{D}_2$ be two neighboring datasets. Denote by $\mathcal{L}^{SaS}_{\mathcal{D}_1|| \mathcal{D}_2}(x; \gamma)$ the privacy loss of observation $x$ for a bounded query $f$ perturbed by a SaS Mechanism $\mathcal{M}_f$ with scale parameter $\gamma$. In the limit as $|x|$ tends to $\infty$, the behavior of the privacy loss is indistinguishably asymptotic for distinct choices of $\gamma$:
\begin{equation}
    \lim_{|x| \to \infty} \mathcal{L}^{SaS}_{\mathcal{D}_1|| \mathcal{D}_2}(x; \gamma_1) = \lim_{|x| \to \infty} \mathcal{L}^{SaS}_{\mathcal{D}_1|| \mathcal{D}_2} (x; \gamma_2),
\end{equation}
for $\gamma_1 \neq \gamma_2$.
\end{lemma}
\begin{proof}
The proof follows directly the limit as $|x|$ is taken to $\infty$ in equation (\ref{eqn:sas}) with the substitutions $\hat{t} = \gamma t$ and $\hat{x} = x / \gamma$.
    % \begin{equation}
    % \label{eqn:gamma_1}
    % \begin{aligned}       
    %     p(x; \alpha, \gamma, \mu) & =  \frac{1}{2\gamma\pi}\int_{-\infty}^{\infty} e^{-| \hat{t}|^\alpha-i(\hat{x}-\mu_i)\hat{t}}d\hat{t}\\
    %     & = p(\hat{x}; \alpha, 1, \mu).
    % \end{aligned}
    % \end{equation}
    % Substituting (\ref{eqn:gamma_1}) into the privacy loss function (\ref{eqn:plf}) gives
    % \begin{equation}
    %     \begin{aligned}       
    %     \mathcal{L}^{SaS}_{\mathcal{D}_1|| \mathcal{D}_2}(x; \gamma) & = \ln \frac{\int_{-\infty}^{\infty} e^{-| \hat{t}|^\alpha-i(\hat{x}-f(\mathcal{D}_1))\hat{t}}d\hat{t}}{\int_{-\infty}^{\infty} e^{-| \hat{t}|^\alpha-i(\hat{x}-f(\mathcal{D}_2))\hat{t}}d\hat{t}} \\
    %     & = \mathcal{L}^{SaS}_{\mathcal{D}_1|| \mathcal{D}_2}(\hat{x}; 1).
    %     \end{aligned}
    % \end{equation}
    % Observing that $|\hat{x}|$ tends to $\infty$ as $|x|$ is driven to $\infty$ we have
    % \begin{equation}
    %     \lim_{|x|\to\infty} \mathcal{L}^{SaS}_{\mathcal{D}_1|| \mathcal{D}_2}(x; \gamma) = \lim_{|\hat{x}|\to\infty} \mathcal{L}^{SaS}_{\mathcal{D}_1|| \mathcal{D}_2}(\hat{x}; 1), \quad \forall \gamma.
    % \end{equation}
    \vspace{-1mm}
    % and that in the limit as $|x|$ tends to infinity, the shift and scale of $\hat{x}_1$ and $\hat{x}_2$ are irrelevant.
\end{proof}

With the above results, we are now in a position to state and prove our main contribution, namely, that for $\alpha \in [1,2)$, the privacy loss for SaS densities is bounded, i.e. the SaS Mechanism is $\varepsilon$-differentially private.
\begin{theorem}
\label{thm:sas-dp}
   (The SaS Mechanism is $\epsilon$-DP)  
    Let $\mathcal{D}_1 \simeq \mathcal{D}_2$ be two neighboring datasets and let $f$ be a bounded query that operates on them. Consider the SaS Mechanism, which we denote by $\mathcal{M}_f$, with stability parameter $\alpha$ in the reduced range  $\alpha \in [1, 2)$. Then, the mechanism $\mathcal{M}_f$ satisfies (\ref{eqn:pureDP}), $\varepsilon$-DP.
\end{theorem}
\begin{proof}

Each element of the mechanism's output is the perturbation of the queries response by an independent sample from the uni-variate density in (\ref{eqn:sas}). Thus, the joint density is equal to the product of the individual densities. As a result, we can write the privacy loss for a given observation vector $\x$ as 
\begin{equation}
% \label{eqn:sas-plf}
     \mathcal{L}^{SaS}_{\mathcal{D}_1|| \mathcal{D}_2}(\x) = \ln \frac{\prod\limits_{i = 1}^m p_{SaS}(x_i; \alpha, \gamma,f(\mathcal{D}_1)_i)}{\prod\limits_{i = 1}^mp_{SaS}(x_i; \alpha, \gamma,f(\mathcal{D}_2)_i)}.
\end{equation}
This can be written as the sum of the log-ratios:
\begin{equation}
     \mathcal{L}^{SaS}_{\mathcal{D}_1, \mathcal{D}_2}(\x) = \sum_{i=1}^m \ln \frac{p_{SaS}(x_i; \alpha, \gamma,f(\mathcal{D}_1)_i)}{p_{SaS}(x_i; \alpha, \gamma,f(\mathcal{D}_2)_i)}.
\end{equation}
Without loss of generality, let this sum be written in decreasing order of magnitudes of the terms, i.e. the first term, $i=1$, has the largest magnitude. We now have the following bound:
\begin{equation}
\label{eqn:sas-plf}
     \big|\mathcal{L}^{SaS}_{\mathcal{D}_1|| \mathcal{D}_2}(\x)\big| \leq m \Big|\ln \frac{p_{SaS}(x_1; \alpha, \gamma,f(\mathcal{D}_1)_1)}{p_{SaS}(x_1; \alpha, \gamma,f(\mathcal{D}_2)_1)}\Big|.
\end{equation}
Our objective is to prove that the right side of (\ref{eqn:sas-plf}) is bounded as function of $x_1$, which will imply, by (\ref{eqn:purepdf}), the mechanism is $\varepsilon$-differentially private. We do so by first proving that the privacy loss is bounded on any compact set. Note that this is not immediate since we are dealing with the log of a ratio and have no assurance that the numerator or denominator ever vanishes. Then, we show that in the limit as $|x|$ tends to infinity, the privacy loss tends to $0$, and thus does not diverge.

Initially, let $x_1$ be an element in a compact set $[a, b] \subset \mathbb{R}$. The log-ratio of the densities could become unbounded within a finite interval in two ways: the argument vanishes or diverges. Consider first the case where one of densities vanishes within the interval. However, by Lemma \ref{lem:sup}, an SaS density has support on the entire real line, $\mathbb{R}$. Therefore, the density is strictly positive over all compact sets $[a,b] \subset \mathbb{R}$. 

Then, we consider if the numerator or denominator of (\ref{eqn:sas-plf}) could be unbounded within the interval $[a,b]$. For simplicity, let $\mu=0$ and apply the substitution $e^{-ix_1} = \cos(tx_1) - i\sin(tx_1)$ to the representation of the SaS density (\ref{eqn:sas}):
\begin{equation}
    p_{SaS}(x_1;\alpha, \gamma, 0) = \frac{1}{2\pi}\int_{-\infty}^{\infty}e^{-|\gamma t|^\alpha}(\cos(tx_1) - i\sin(tx_1))dt.
\end{equation}
Splitting the integral we have
\begin{equation}
\begin{aligned}    
    p_{SaS}(x_1;\alpha, \gamma, 0) = & \frac{1}{2\pi}\int_{-\infty}^{\infty}e^{-|\gamma t|^\alpha}\cos(tx_1)dt \\ - & i \frac{1}{2\pi}\int_{-\infty}^{\infty}e^{-|\gamma t|^\alpha} \sin(tx_1)dt.
\end{aligned}
\end{equation}
Since $\sin(tx_1)$ is an odd function the second integral vanishes:
\begin{equation}
\label{eqn:cos-form}
    p_{SaS}(x_1;\alpha, \gamma, 0) = \frac{1}{2\pi}\int_{-\infty}^{\infty}e^{-|\gamma t|^\alpha}\cos(tx_1)dt
\end{equation}
As $\cos(tx_1)$ is bounded above by $1$:
\begin{equation}
\label{eqn:a_part}
    p_{SaS}(x_1;\alpha, \gamma, 0) \leq   \frac{1}{2\pi}\int_{-\infty}^{\infty}e^{-|\gamma t|^\alpha}dt
\end{equation}
Observe that the integrand in (\ref{eqn:a_part}) is symmetric about $t=0$ and so we can remove the absolute value:
\begin{equation}
\label{eqn:30}
    p_{SaS}(x_1;\alpha, \gamma, 0) \leq   \frac{1}{\pi}\int_{0}^{\infty}e^{-(\gamma t)^\alpha}dt
\end{equation}
Using the substituting $\hat{t} = (\gamma t)^\alpha$ (\ref{eqn:30}) becomes
\begin{equation}
\label{eqn:31}
\begin{aligned}
    p_{SaS}(x_1;\alpha, \gamma, 0) & \leq   \frac{1}{\alpha\gamma\pi}\int_{0}^{\infty}\hat{t}^{\frac{1}{\alpha}-1}e^{-\hat{t}^\alpha}d\hat{t}\\ 
    & = \frac{\Gamma(\frac{1}{\alpha})}{\alpha \gamma\pi}
\end{aligned}
\end{equation}
where $\Gamma$ is the standard $\Gamma$ function which is finite on the interval $1/\alpha \in (1/2, 1)$ \cite{oeis}. Equation (\ref{eqn:31}) states that the density $p_{SaS}$ is bounded over the real line. It is therefore bounded on the compact subset $[a,b]$. 

Next, we proceed to prove that the privacy loss remains bounded in the limit as $|x_1|$ tends to infinity. 

Recall the series expansion, for $\gamma=1$, presented in Lemma \ref{lem:sum}. Truncate the series to a single term, i.e., $n=1$, and consider the privacy loss after substitution in (\ref{eqn:sas-plf}):
\begin{equation}
    \big|\mathcal{L}^{SaS}_{\mathcal{D}_1||\mathcal{D}_2}(\x)\big| \leq m\Big|\ln\frac{\big(x_1 - f(\mathcal{D}_1)\big)^{-\alpha-1} + O(x_1^{-2\alpha - 1})}{\big(x_1 - f(\mathcal{D}_2)\big)^{-\alpha-1}  + O(x_1^{-2\alpha - 1})}\Big|.
\end{equation}
Thus, in the limit as $|x_1|$ tends infinity, the error terms in the numerator and denominator are dominated by the first terms:
\begin{equation}
\begin{aligned}
   \lim_{||\x||\to\infty} & \big|\mathcal{L}^{SaS}_{\mathcal{D}_1||\mathcal{D}_2}(\x)\big| \leq \\ \lim_{|x_1|\to\infty} & m\Big|\ln\frac{\big(x_1 - f(\mathcal{D}_1)\big)^{-\alpha-1} + O(x_1^{-2\alpha - 1})}{\big(x_1 - f(\mathcal{D}_2)\big)^{-\alpha-1}  + O(x_1^{-2\alpha - 1})}\Big| = \\
   \lim_{|x_1|\to\infty} & m\Big|\ln\frac{\big(x_1 - f(\mathcal{D}_1)\big)^{-\alpha-1}}{\big(x_1 - f(\mathcal{D}_2)\big)^{-\alpha-1}}\Big| =
   0.
\end{aligned}
\end{equation}
Thus, the privacy loss converges to $0$. By Lemma \ref{lem:scale}, the choice of $\gamma$ does not impact the asymptotic behavior. Since this result holds for any value of $\x \in \mathcal{R}(\mathcal{M}_f)$, by Eq. (\ref{eqn:purepdf}), we have proved that the SaS Mechanism is $\varepsilon$-DP.
\end{proof}
We next establish a measure of the error that a privacy mechanism introduces.

\section{Expected Error of SaS Mechanism}
\label{sec:error}
It is common for methods to use the $\ell_2$-norm in defining such an error measure. However, the moment of the SaS densities is only defined up to $\alpha$, and since we are considering $\alpha < 2$, the second moment is not well defined \cite{nolan20}. In lieu of the $\ell_2$-norm, we choose to use the mean absolute deviation (MAD), also used in \cite{roth14}:
\begin{definition}
    (Expected Privacy Distortion) Denote by $f(\mathcal{D})$ and $\mathcal{M}_f(\mathcal{D})$ the response of the query and privacy mechanism respectively. The mean absolute deviation is
    \begin{equation}
        E(f(\mathcal{D}), \mathcal{M}_f(\mathcal{D})) := \mathbb{E} |f(\mathcal{D}) - \mathcal{M}_f(\mathcal{D})|,
    \end{equation}
    which is equivalent to the expectation of the absolute value of the injected noise $Y$:
    \begin{equation}
        E(f(\mathcal{D}), \mathcal{M}_f(\mathcal{D})) = \mathbb{E} |Y|.
    \end{equation}
    \vspace{-7mm}
\begin{equation*}\tag*{\textrm{$\blacktriangleleft$}}\end{equation*}
% \vspace{-3mm}
\end{definition}

Before we can study the error incurred under the SaS Mechanism, we need the fact that the SaS Mechanism is strictly stable.
\begin{lemma}
\label{lem:strict}
(SaS density is \textit{Strictly} Stable) The SaS density (\ref{eqn:sas}) with location parameter $\mu=0$ is strictly stable.
\end{lemma}
\begin{proof}
For the sake of brevity we omit the proof.
    % Let $Y_1$, $Y_2$, and $Y$ be i.i.d. SaS densities with $\mu=0$. Let $a$ and $b$ be two scalar values and consider the density of the combined random variable $aY_1 + bY_2$. Because the SaS densities are described by their characteristic functions we have the following relation
    % \begin{equation}
    %     \varphi_{aY_1+bY_2}(t) = \varphi_{aY_1}(t) \varphi_{bY_2}(t).
    % \end{equation}
    % Using the definition of characteristic function, we bring the constants into the argument
    % \begin{equation}
    % \begin{aligned}
    %     \varphi_{aY_1}(t) \varphi_{bY_2}(t) & = \mathbb{E}[e^{itaY_1}]\mathbb{E}[e^{itbY_2}] \\ 
    %     & = \varphi_{Y_1}(at) \varphi_{Y_2}(bt)
    % \end{aligned}
    % \end{equation}
    % Expand by substituting the expression for the characteristic function of a stable distribution with $\mu=0$  (\ref{eqn:char}) into both functions on the right side,
    %  \begin{equation}
    % \begin{aligned}
    %     \varphi_{Y_1}(at) \varphi_{Y_2}(bt) & = \exp(|\gamma a t|^\alpha)\exp(|\gamma bt|^\alpha)\\
    %     & =  \exp{ |(a^\alpha+b^\alpha)^{1/\alpha}\gamma t|^\alpha}
    % \end{aligned}
    % \end{equation}
    % Setting $c = (a^\alpha + b^\alpha)^{1/\alpha}$ gives $aY_1 + bY_2 = cY$.
\end{proof}
We are now equipped to determine the expected error introduced in the query by the SaS Mechanism.
\begin{theorem}
\label{thm:distortion}
    (Expected Distortion Due to SaS Mechanism) Let $f$ be a bounded query that operates on dataset $\mathcal{D}$. Denote by $\mathcal{M}_f$ the SaS Mechanism and take the stability parameter $\alpha$ to be restricted to the range $\alpha\in(1,2)$. Then, the mean absolute distortion is
    \begin{equation}
        \label{eqn:mad-sas}
        E\big(f(\mathcal{D}, \mathcal{M}_f(\mathcal{D})\big) = \frac{2\gamma}{\pi}\Gamma\big(1-\frac{1}{\alpha}\big).
    \end{equation}
\end{theorem}
\begin{proof}
    In \cite{nolan20}, the proof of Corollary 3.5 includes a statement that if a density $p_Y$ is strictly stable then (\ref{eqn:mad-sas}) holds.
\end{proof}

We now provide the expected distortions of the two most common privacy mechanisms: the Laplace and the Gaussian mechanisms \cite{roth14, dwork06b}, to show that each induces an error linear in the scale of the noise. The mean absolute deviation of the Laplace density is
\begin{equation}
    \mathbb{E}[|Lap(0, b)|] = \mathbb{E} [Exp(b^{-1})] = b.
\end{equation}
The mean absolute deviation Gaussian density is the expected value of the half-normal random variable 
\begin{equation}
    \mathbb{E}[|\mathcal{N}(0, \sigma^2)|] = \sqrt{\frac{2}{\pi}}\sigma.
\end{equation}
Note that for each of the three densities the error is related linearly to the density's respective scale. From (\ref{eqn:mad-sas}) we recover the distortion of the Gaussian mechanism by taking $\alpha=2$ and $\gamma= \sigma/\sqrt{2}$.
% Next, we proceed to prove that the expected distortion is monotonic in $\alpha$, reaching a minimum when $\alpha=2$ and diverging as $\alpha$ tends to $1$.

% \begin{remark} We note that the relationship between the expected error of the SaS Mechanism $\mathbb{E}[|Y^{SaS}|]$ and the scale of the density $\gamma$ in (\ref{eqn:mad-sas}) is linear. In other words, the level of privacy is inversely related to the accuracy of statistics base on mechanism $\mathcal{M}_f$. We note that this relationship is shared by other common mechanisms such as the Laplace and Gaussian.
% \end{remark}

% \begin{corollary}
% \label{cor:mon-error}
%     (Error is Monotonic in $\alpha$) The mean absolute distortion injected into a query by the SaS Mechanism decreases monotonically as $\alpha$ increases from $1$ to $2$.
% \end{corollary}
% \begin{proof}
% Because $\alpha$ is bound within $(1,2)$, the argument of the Gamma function in (\ref{eqn:mad-sas}) varies between $(0,1/2)$. The Gamma function has an asymptote at $x=0$ and reaches a local minimum in the right plane at $x \approx 1.462$ \cite{oeis}. Thus, for a given $\gamma$, the distortion in Eq. (\ref{eqn:mad-sas}) is minimized when $\alpha$ tends to $2$.
% \end{proof}

% A naive first thought is that $\alpha=2$ is the superior value for the parameter. However, recall that when $\alpha=2$, the SaS Mechanism becomes the Gaussian and no longer satisfies pure differential privacy. Indeed, the problem of designing a mechanism for a given problem is equivalent to a weighting between the privacy budget and the induced error. 

\section{CONCLUSION}
\label{sec:conclusion}
In conclusion, the SaS Mechanism represents a significant advancement in the field of differential privacy. This mechanism not only provides strong guarantees of privacy but also offers distinct advantages when compared to other common privacy mechanisms. Because the SaS Mechanism is closed under convolution makes it a particularly good choice for applications seeking to implement local differential privacy, such as federated learning. Looking forward, we are actively investigating the privacy of the SaS Mechanism under other versions of Differential Privacy such as Renyi and Concentrated DP.

\section*{ACKNOWLEDGMENT}
This work was supported in part by Grant Number ECCS-2024493 from the U.S. National Science Foundation.

\bibliographystyle{unsrt}
\bibliography{root}{}

\end{document}